\documentclass[a4paper, 10pt]{article}
\usepackage[margin=2cm]{geometry}
\usepackage[utf8]{inputenc}
\usepackage[english]{babel}
\usepackage{amssymb}
\usepackage{amsmath}
\usepackage{amsthm}
\usepackage{titlesec}
\usepackage{mathtools}
\usepackage{algorithm}
\usepackage[noend]{algpseudocode}
\usepackage[normalem]{ulem}
\usepackage[round]{natbib}
\usepackage{float}
\usepackage{svg}
\usepackage{fancyhdr}
\usepackage[hidelinks]{hyperref}
\usepackage[textsize=small, textwidth=1.5cm]{todonotes}
\usepackage{thmtools}
\usepackage{thm-restate}

\mathtoolsset{showonlyrefs}

\newcommand{\algrule}[1][.2pt]{\par\vspace\baselineskip\hrule height #1\par\vspace\baselineskip}

\DeclareMathOperator{\on}{on}
\DeclareMathOperator{\off}{off}
\DeclareMathOperator{\idle}{idle}
\DeclareMathOperator{\busy}{busy}
\DeclareMathOperator{\costs}{costs}

\DeclareMathOperator{\OPT}{OPT}
\DeclareMathOperator{\opt}{opt}
\DeclareMathOperator{\pltr}{pltr}
\DeclareMathOperator{\PLTR}{PLTR}
\DeclareMathOperator{\ltr}{LTR}
\DeclareMathOperator{\LTR}{LTR}
\DeclareMathOperator{\aug}{aug}
\DeclareMathOperator{\real}{real}
\DeclareMathOperator{\fv}{fv}
\DeclareMathOperator{\uv}{uv}
\DeclareMathOperator{\pv}{pv}
\DeclareMathOperator{\vol}{vol}
\DeclareMathOperator{\opdef}{def}
\DeclareMathOperator{\exc}{exc}
\DeclareMathOperator{\keepidle}{keepidle}
\DeclareMathOperator{\keepbusy}{keepbusy}

\DeclareMathOperator{\crit}{crit}
\DeclareMathOperator{\fillop}{fill}
\DeclareMathOperator{\close}{close}
\DeclareMathOperator{\res}{Sup}

\newtheorem{theorem}{Theorem}
\newtheorem{lemma}[theorem]{Lemma}
\newtheorem{definition}[theorem]{Definition}

\title{Greedy Minimum-Energy Scheduling}
\author{Gunther Bidlingmaier\thanks{Department of Computer Science, Technical University of Munich, Germany,\ \text{g.bidlingmaier@tum.de}}}
\date{}
\begin{document}

\selectlanguage{english}

\maketitle
\begin{abstract}
We consider the problem of energy-efficient scheduling across multiple processors with a power-down mechanism.
In this setting a set of $n$ jobs with individual release times, deadlines, and processing volumes must be scheduled across $m$ parallel processors while minimizing the consumed energy.
Idle processors can be turned off to save energy, while turning them on requires a fixed amount of energy.
For the special case of a single processor, the greedy Left-to-Right algorithm \citep{irani_left_to_right_soda_2003} guarantees an approximation factor of $2$.
We generalize this simple greedy policy to the case of $m \geq 1$ processors running in parallel and show that the energy costs are still bounded by $2 \OPT + P$, where $\OPT$ is the energy consumed by an optimal solution and $P < \OPT$ is the total processing volume.
Our algorithm has a running time of $\mathcal{O}(n f \log d)$, where $d$ is the difference between the latest deadline and the earliest release time, and $f$ is the running time of a maximum flow calculation in a network of $\mathcal{O}(n)$ nodes.
\end{abstract}
\clearpage


\section{Introduction}
Energy-efficiency has become a major concern in most areas of computing for reasons that go beyond the apparent ecological ones.
At the hardware level, excessive heat generation from power consumption has become one of the bottlenecks.
For the billions of mobile battery-powered devices, power consumption determines the length of operation and hence their usefulness.
On the level of data centers, electricity is often the largest cost factor and cooling one of the major design constraints.
Algorithmic techniques for saving power in computing environments employ two fundamental mechanisms, first the option to power down idle devices and second the option to trade performance for energy-efficiency by speed-scaling processors.
In this paper we study the former, namely classical deadline based scheduling of jobs on parallel machines which can be powered down with the goal of minimizing the consumed energy.

In our setting, a computing device or processor has two possible states, it can be either \emph{on} or \emph{off}.
If a processor is on, it can perform computations while consuming energy at a fixed rate.
If a processor is off, the energy consumed is negligible but it cannot perform computation.
Turning on a processor, i.e.\ transitioning it from the off-state to on-state consumes additional energy.
The problem we have to solve is to schedule a number of jobs or tasks, each with its own processing volume and interval during which it has to be executed.
The goal is to complete every job within its execution interval using a limited number of processors while carefully planning idle times for powering off processors such that the consumed energy is minimized.
Intuitively, one aims for long but few idle intervals, so that the energy required for transitioning between the states is low, while avoiding turned on processors being idle for too long.

\paragraph{Previous work}
This fundamental problem in power management was first considered by~\cite{irani_left_to_right_soda_2003} for a single processor.
In their paper, they devise arguably the simplest algorithm one can think of which goes beyond mere feasibility.
Their greedy algorithm \emph{Left-to-Right} ($\ltr$) is a $2$-approximation and proceeds as follows.
If the processor is currently busy, i.e.\ working on a job, then $\ltr$ greedily keeps the processor busy for as long as possible, always working on the released job with the earliest deadline.
Once there are no more released jobs to be worked on, the processor becomes idle and $\ltr$ greedily keeps the processor idle for as long as possible such that all remaining jobs can still be feasibly completed.
At this point, the processor becomes busy again and $\ltr$ proceeds recursively until all jobs are completed.

The first optimal result for the case of a single processor and jobs with unit processing volume was developed by~\cite{baptiste_unit_jobs}.
He devised a dynamic program that runs in time $\mathcal{O}(n^7)$, where $n$ denotes the number of jobs to be scheduled.
Building on this result,~\cite{baptiste_general_jobs} solved the case of general processing volumes on a single processor in time $\mathcal{O}(n^5)$.
Their sophisticated algorithm involves the computation of multiple dynamic programming tables, the introduction of a special method for speeding up the computation of these tables, and a final post-processing phase.

The first result for an arbitrary number of processors $m$ was given by~\cite{demaine} for the special case of unit processing volumes.
They solved this special case in time $~\mathcal{O}(n^7m^5)$ by building on the original dynamic programming approach of~\cite{baptiste_unit_jobs} while non-trivially obtaining additional structure.
Obtaining good solutions for general job weights is difficult because of the additional constraint that every job can be worked on by at most a single processor at the same time.
Note that this is not an additional restriction for the former special case of unit processing volumes since time is discrete in our problem setting.
It is a major open problem whether the general multi-processor setting is NP-hard.
It took further thirteen years for the first non-trivial result on the general setting to be be developed, i.e.\ an algorithm for the case of multiple processors and general processing volumes of jobs.
In their breakthrough paper,~\cite{antoniadis} develop the first constant-factor approximation for the problem.
Their algorithm guarantees an approximation factor of $3 + \epsilon$ and builds on the Linear Programming relaxation of a corresponding Integer Program.
Their algorithm obtains a possibly infeasible integer solution by building the convex hull of the corresponding fractional solution.
Since this integer solution might not schedule all jobs, they develop an additional \textit{extension} algorithm \textit{EXT-ALG}, which iteratively extends the intervals returned by the rounding procedure by a time slot at which an additional turned on processor allows for an additional unit of processing volume to be scheduled.

\cite{skeletons} improve this approximation factor to $2 + \epsilon$ by incorporating into the Linear Program additional constraints for the number of processors required during every possible time interval.
They also modify the rounding procedure based on their concept of a \textit{multi-processor skeleton}.
Very roughly, a skeleton is a stripped-down schedule which still guarantees a number of processors in the on-state during specific intervals and which provides a lower bound for the costs of an optimal feasible schedule.

Building on this concept of skeletons, they also develop a combinatorial $6$-approximation for the problem.
This algorithm first computes the lower bounds for the number of processors required in every possible time interval starting and ending at a release time or deadline using flow calculations.
Based on these bounds, they define for every processor a single-processor scheduling problem with $\mathcal{O}(n^2)$ artificial jobs.
For each of these single-processor problems they construct a single-processor skeleton using dynamic programming.
These in turn are then combined into a multi-processor skeleton, which is extended into a feasible schedule by first executing EXT-ALG, and then carefully powering on additional processors since EXT-ALG is not sufficient for ensuring feasibility here.

As presented in the papers, both Linear Programs of~\cite{antoniadis} and~\cite{skeletons}, respectively, run in pseudo-polynomial time.
By using techniques presented in~\cite{antoniadis}, the number of time slots which have to be considered can be reduced from $d$ to $\mathcal{O}(n \log d)$, allowing the algorithms to run in polynomial time.
More specifically, the number of constraints and variables of the Linear Programs reduces to $\mathcal{O}(n^2 \log^2 d)$.
However, this improved running time comes at the price of the additive $\epsilon$ in the approximation factors of the two LP-based algorithms.
The running time of the EXT-ALG used by all three approximation algorithms is reduced to $\mathcal{O}(F m n^3 \log^3 d)$, where $F$ refers to a maximum flow calculation in a network with $\mathcal{O}(n^2 \log d)$ edges and $\mathcal{O}(n \log d)$ nodes.

\paragraph{Contribution}
In this paper we develop a greedy algorithm which is simpler and faster than the previous algorithms.
The initially described greedy algorithm Left-to-Right of~\cite{irani_left_to_right_soda_2003} is arguably the simplest algorithm one can think of for a single processor.
We naturally extend $\LTR$ to multiple processors and show that this generalization still guarantees a solution of costs at most $2 \OPT + P$, where $P < \OPT$ is the total processing volume.
Our simple greedy algorithm \textit{Parallel Left-to-Right} ($\PLTR$) is the combinatorial algorithm with the best approximation guarantee and does not rely on Linear Programming and the necessary rounding procedures of~\cite{antoniadis} and~\cite{skeletons}.
It also does not require the EXT-ALG, which all previous algorithms rely on to make their infeasible solutions feasible in an additional phase.

Indeed, $\PLTR$ only relies on the original greedy policy of Left-to-Right: just keep processors in their current state (busy or idle) for as long as feasibly possible.
For a single processor, $\LTR$ ensures feasibility by scheduling jobs according to the policy Earliest-Deadline-First (EDF).
For checking feasibility if multiple processors are available, a maximum flow calculation is required since EDF is not sufficient anymore.
Correspondingly, our generalization $\PLTR$ uses such a flow calculation for checking feasibility.

While the $\PLTR$ algorithm we describe in Section~\ref{section:algorithm} is very simple, the structure exhibited by the resulting schedules is surprisingly rich.
This structure consists of \textit{critical sets of time slots} during which $\PLTR$ only schedules the minimum amount of volume which is feasibly possible.
In Section~\ref{section:structure} we show that whenever $\PLTR$ requires an additional processor to become busy at some time slot $t$, there must exist a critical set of time slots containing $t$.
This in turn gives a lower bound for the number of busy processors required by any solution.

Devising an approximation guarantee from this structure is however highly non-trivial and much more involved than the approximation proof of the single-processor $\LTR$ algorithm, because one has to deal with sets of time slots and not just intervals.
Our main contribution in terms of techniques is a complex procedure which (for the sake of the analysis only) carefully realigns the jobs scheduled in between critical sets of time slots such that it is sufficient to consider intervals as in the single processor case, see Section~\ref{section:approximation} for details.

Finally, we show in Section~\ref{section:running_time} that the simplicity of the greedy policy also leads to a much faster algorithm than the previous ones, namely to a running time $\mathcal{O}(n f \log d)$, where $d$ is the maximal deadline and $f$ is the running time for checking feasibility by finding a maximum flow in a network with $\mathcal{O}(n)$ nodes.

\paragraph{Formal Problem Statement}
Formally, a problem instance consists of a set $J$ of jobs with an integer release time $r_j$, deadline $d_j$, and processing volume $p_j$ for every job $j \in J$.
Each job $j \in J$ has to be scheduled across $m \geq 1$ processors for $p_j$ units of time in the execution interval $E_j \coloneqq [r_j, d_j]$ between its release time and its deadline.
Preemption of jobs and migration between processors is allowed at discrete times and occurs without delay, but no more than one processor may process any given job at the same time.
Without loss of generality, we assume the earliest release time to be $0$ and denote the last deadline by $d$.
The set of discrete time slots is denoted by $T \coloneqq \{0, \ldots, d\}$.
The total amount of processing volume is $P \coloneqq \sum_{j \in J} p_j$.

Every processor is either completely off or completely on in every discrete time slot $t \in T$.
A processor can only work on some job in the time slot $t$ if it is in the on-state.
A processor can be turned on and off at discrete times without delay.
All processors start in the off-state.
The objective now is to find a feasible schedule which minimizes the expended energy $E$, which is defined as follows.
Each processor consumes $1$ unit of energy for every time slot it is in the on-state and $0$ units of energy if it is in the off-state.
Turning a processor on consumes a constant amount of energy $q \geq 0$, which is fixed by the problem instance.
In Graham's notation~\citep{graham}, this setting can be denoted with $m \mid r_j; \overline{d_j}; \mathrm{pmtn} \mid E$.

\paragraph{Costs of busy and idle intervals}
We say a processor is \emph{busy} at time $t \in T$ if some job is scheduled for this processor at time $t$.
Otherwise, the processor is \emph{idle}.
Clearly a processor cannot be busy and off at the same time.
An interval $I \subseteq T$ is a (full) \emph{busy interval} for processor $k \in [m]$ if $I$ is inclusion maximal on condition that processor $k$ is busy in every $t \in I$.
Correspondingly, an interval $I \subseteq T$ is a \emph{partial busy interval} for processor $k$ if $I$ is not inclusion maximal on condition that processor $k$ is busy in very $t \in I$.
We define (partial and full) \emph{idle intervals}, \emph{on intervals}, and \emph{off intervals} of a processor analogously via inclusion maximality.
Observe that if a processor is idle for more than $q$ units of time, it is worth turning the processor off during the corresponding idle interval.
Our algorithm will specify for each processor when it is busy and when it is idle.
Each processor is then defined to be in the off-state during idle intervals of length greater than $q$ and otherwise in the on-state.
Accordingly, we can express the costs of a schedule $S$ in terms of busy and idle intervals.

For a multi-processor schedule $S$, let $S^k$ denote the schedule of processor $k$.
Furthermore, for fixed $k$, let $\mathcal{N}, \mathcal{F}, \mathcal{B}, \mathcal{I}$ be the set of on, off, busy, and idle intervals on $S^k$. 
We partition the costs of processor $k$ into the costs $\on(S^k)$ for residing in the on-state and the costs $\off(S^k)$ for transitioning between the off-state and the on-state, hence $\costs(S^k) = \on(S^k) + \off(S^k) = \sum_{N \in \mathcal{N}} |N| + q$.
Equivalently, we partition the costs of processor $k$ into the costs $\idle(S^k) \coloneqq \sum_{I \in \mathcal{I}} \min \{ |I|, q \}$ for being idle and the costs $\busy(S^k) \coloneqq \sum_{B \in \mathcal{B}} |B|$ for being busy.
The total costs of a schedule $S$ are the total costs across all processors, i.e.\ $\costs(S) = \sum_{k = 1}^{m} \costs(S^k)$.
Clearly we have $\sum_{k = 1}^m \busy(k) = P$, this means for an approximation guarantee the critical part is bounding the idle costs.

\paragraph{Lower and upper bounds for the number of busy processors}
We specify a generalization of our problem which we call \emph{deadline-scheduling-with-processor-bounds}.
Where in the original problem, for each time slot $t$, between $0$ and $m$ processors were allowed to be working on jobs, i.e.\ being busy, we now specify a lower bound $l_t \geq 0$ and an upper bound $m_t \leq m$.
For a feasible solution to \emph{deadline-scheduling-with-processor-bounds}, we require that in every time slot $t$, the number of busy processors, which we denote with $\vol(t)$, lies within the lower and upper bounds, i.e.\ $l_t \leq v(t) \leq m_t$.
This will allow us to express the $\PLTR$ greedy policy of keeping processors idle or busy, respectively.
Note that this generalizes the problem \emph{deadline-scheduling-on-intervals} introduced by~\cite{antoniadis} by additionally introducing lower bounds.

\paragraph{Properties of an optimal schedule}
\begin{definition}\label{def:stair_property}
  Given some arbitrary but fixed order on the number of processors, a schedule $S$ fulfills the \emph{stair-property} if it uses the lower numbered processors first, i.e.\ for every $t \in T$, if processor $k \in [m]$ is busy at $t$, then every processor $k' \leq k$ is busy at $t$.
  This symmetrically implies that if processor $k \in [m]$ is idle at $t$, then every processor $k' \geq k$ is idle at $t$.
\end{definition}

\begin{restatable}{lemma}{lemmastairproperty}\label{lemma:stair_property_opt}
  For every problem instance we can assume the existence of an optimal schedule $S_{\opt}$ which fulfills the stair-property.
\end{restatable}

\section{Algorithm}\label{section:algorithm}

The \emph{Parallel Left-to-Right} ($\PLTR$) algorithm shown in Algorithm~\ref{alg:pltr} iterates through the processors in some arbitrary but fixed order and keeps the current processor idle for as long as possible such that the scheduling instance remains feasible.
Once the current processor cannot be kept idle for any longer, it becomes busy and $\PLTR$ keeps it and all lower-numbered processors busy for as long as possible while again maintaining feasibility.
The algorithm enforces these restrictions on the busy processors by iteratively making the lower and upper bounds $l_t$, $m_t$ of the corresponding instance of \emph{deadline-scheduling-with-processor-bounds} more restrictive.
Visually, when considering the time slots on an axis from left to right and when stacking the schedules of the individual processors on top of each other, this generalization of the single processor \emph{Left-to-Right} algorithm hence proceeds \emph{Top-Left-to-Bottom-Right}.

Once $\PLTR$ returns with the corresponding tight upper and lower bounds $m_t$, $l_t$, an actual schedule $S_{\pltr}$ can easily be constructed by running the flow-calculation used for the feasibility check depicted in Figure~\ref{fig:flow} or just taking the result of the last flow-calculation performed during $\PLTR$.
The mapping from this flow to an actual assignment of jobs to processors and time slots can then be defined as described in Lemma~\ref{lemma:flow_feasibility}, which also ensures that the resulting schedule fulfills the stair-property from Definition~\ref{def:stair_property}, i.e.\ that it always uses the lower-numbered processors first.

\begin{algorithm}
  \caption{Parallel Left-to-Right}\label{alg:pltr}
\begin{algorithmic}
  \State{$m_t \gets m$ for all $t \in T$}
  \State{$l_t \gets 0$ for all $t \in T$}
  \For{$k \gets m$ \textrm{to} $1$}
    \State{$t\gets 0$}
    \While{$t \leq d$}
      \State{$t \gets $\Call{$\keepidle$}{$k, t$}}
      \State{$t \gets $\Call{$\keepbusy$}{$k, t$}}
    \EndWhile{}
  \EndFor{}
\algrule
  \Function{$\keepidle$}{$k, t$}
    \State{search maximal $t' > t$ s.t.\
    exists feasible schedule with $m_{t''}$ set to $k-1$ for all $t'' \in [t, t')$}
    \State{$m_{t''} \gets k - 1$ for all $t'' \in [t, t')$}
    \State{\Return{$t'$}}
  \EndFunction{}
  \Function{$\keepbusy$}{$k, t$}
    \State{search maximal $t' > t$ s.t.\
    exists feasible schedule with $l_{t''}$ set to $\max\{k, l_{t''}\}$ for all $t'' \in [t, t')$}
    \State{$l_{t''} \gets \max\{k, l_{t''}\}$ for all $t'' \in [t, t')$}
    \State{\Return{$t'$}}
  \EndFunction{}
\end{algorithmic}
\end{algorithm}
As stated in Lemma~\ref{lemma:flow_feasibility}, the check for feasibility in subroutines $\keepidle$ and $\keepbusy$ can be performed by calculating a maximum $\alpha$-$\omega$ flow in the flow network given in Figure~\ref{fig:flow} with a node $u_j$ for every job $j \in J$ and a node $v_t$ for every time slot $t \in T$ including the corresponding incoming and outgoing edges.

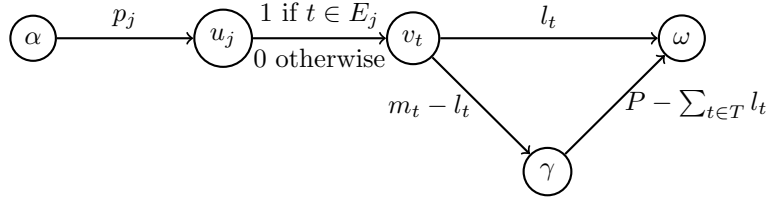
\begin{figure}
  \centering
  \begin{tikzpicture}[node distance={25mm}, thick, main/.style = {draw, circle}]
    \node[main] (1) {$\alpha$};
    \node[main] (2) [right of=1] {$u_j$};
    \node[main] (3) [right of=2] {$v_t$};
    \node[main] (4) [below right of=3] {$\gamma$};
    \node[main] (5) [above right of=4] {$\omega$};
    \draw[->] (1) -- node[midway, above] {$p_j$} (2);
    \draw[->] (2) -- node[midway, above] {$1$ if $t \in E_j$} node[midway, below] {$0$ otherwise} (3);
    \draw[->] (3) -- node[midway, left] {$m_t - l_t$} (4);
    \draw[->] (3) -- node[midway, above] {$l_t$} (5);
    \draw[->] (4) -- node[midway, right] {$P - \sum_{t \in T} l_t$} (5);
  \end{tikzpicture}
  \caption{The Flow-Network for checking feasibility of an instance of \emph{deadline-scheduling-with-processor-bounds} $l_t$ and $m_t$ for the number of busy processors at $t \in T$.
  There is a node $u_j$, $v_t$ with the corresponding edges for every job $j \in J$ and for every time slot $t \in T$, respectively.
  }\label{fig:flow}
\end{figure}

\begin{restatable}{lemma}{lemmaflowfeasibility}\label{lemma:flow_feasibility}
  There exists a feasible solution to an instance of deadline-scheduling-with-processor-bounds $l_t$, $m_t$ if and only if the maximum $\alpha$-$\omega$ flow in the corresponding flow network depicted in Figure~\ref{fig:flow} has value $P$.
\end{restatable}
\begin{theorem}\label{theorem:feasibility}
  Given a feasible problem instance, algorithm $\PLTR$ constructs a feasible schedule.
\end{theorem}
\begin{proof}
  By definition of subroutines $\keepidle$ and $\keepbusy$, $\PLTR$ only modifies the upper and lower bounds $m_t$, $l_t$ for the number of busy processors such that the resulting instance of \emph{deadline-scheduling-with-processor-bounds} remains feasible.
  The correctness of the algorithm then follows from the correctness of the flow-calculation for checking feasibility, which is implied by Lemma~\ref{lemma:flow_feasibility}.
\end{proof}

\section{Structure of the PLTR-Schedule}\label{section:structure}
\subsection{Types of Volume}
\begin{definition}
  For a schedule $S$, a job $j \in J$, and a set $Q \subseteq T$ of time slots, we define
  \begin{enumerate}
    \item
      the \emph{volume} $\vol_S(j, Q)$ as the number of time slots of $Q$ for which $j$ is scheduled by $S$,
    \item
      the \emph{forced volume} $\fv(j, Q)$ as the minimum number of time slots of $Q$ for which $j$ has to be scheduled in every feasible schedule, i.e.\ $\fv(j, Q) \coloneqq \max\{0; p_j - |E_j \setminus Q|\}$,
    \item
      the \emph{unnecessary volume} $\uv_S(j, Q)$ as the amount of volume which does not have to scheduled during $Q$, i.e.\ $\uv_S(j, Q) \coloneqq \vol_S(j, Q) - \fv(j,Q)$,
    \item
      the \emph{possible volume} $\pv(j, Q)$ as the maximum amount of volume which $j$ can be feasibly scheduled in $Q$, i.e.
    $\pv(j, Q) \coloneqq \min\{ p_j, | E_j \cap Q | \}$.
  \end{enumerate}
\end{definition}
Since the corresponding schedule $S$ will always be clear from context, we omit the subscript for $\vol$ and $\uv$.
We extend our volume definitions to sets $J' \subseteq J$ of jobs by summing over all $j \in J'$, i.e.
$\vol(J', Q) \coloneqq \sum_{j \in J'} \vol(j, Q)$.
If the first parameter is omitted, we refer to the whole set $J$, i.e.\ $\vol(Q) \coloneqq \vol(J, Q)$.
For single time slots, we omit set notation, i.e.\ $\vol(t) \coloneqq \vol(J, \{t\})$.
Clearly we have for every feasible schedule, every $Q \subseteq T, j \in J$ that $\fv(j, Q) \leq \vol(j, Q) \leq \pv(j, Q)$.
The following definitions are closely related to these types of volume.
\begin{definition}
  Let $Q \subseteq T$ be a set of time slots.
  We define
  \begin{enumerate}
    \item
      the \emph{density} $\phi(Q) \coloneqq \fv(J, Q) / |Q|$ as the average amount of processing volume which has to be completed in every slot of $Q$,
    \item
      the \emph{peak density} $\hat \phi(Q) \coloneqq \max_{Q' \subseteq Q} \phi(Q')$,
    \item
      the \emph{deficiency} $\opdef(Q) \coloneqq \fv(Q) - \sum_{t \in Q} m_t$ as the difference between the amount of volume which has to be completed in $Q$ and the processing capacity available in $Q$,
    \item
      the \emph{excess} $\exc(Q) \coloneqq \sum_{t \in Q} l_t - \pv(Q)$ as the difference between the processor utilization required in $Q$ and the amount of work available in $Q$.
  \end{enumerate}
\end{definition}
If $\hat \phi(Q) > k - 1$, then clearly at least $k$ processors are required in some time slot $t \in Q$ for every feasible schedule.
If $\opdef(Q) > 0$ or $\exc(Q) > 0$ for some $Q \subseteq T$, then the problem instance is clearly infeasible.

\subsection{Critical Sets of Time Slots}
The following Lemma~\ref{lemma:critical} provides the crucial structure required for the proof of the approximation guarantee.
Intuitively, it states that whenever $\PLTR$ requires processor $k$ to become busy at some time slot $t$, there must be some critical set $Q \subseteq T$ of time slots during which the volume scheduled by $\PLTR$ is minimal.
This in turn implies that processor $k$ needs to be busy at some point during $Q$ in every feasible schedule.
The auxiliary Lemmas~\ref{lemma:cut} and~\ref{lemma:feasibility} provide a necessary and more importantly also sufficient condition for the feasibility of an instance of \emph{deadline-scheduling-with-processor-bounds} based on the excess $\exc(Q)$ and the deficiency $\opdef(Q)$ of sets $Q \subseteq T$.
Lemmas~\ref{lemma:cut} and~\ref{lemma:feasibility} are again a generalization of the corresponding feasibility characterization in~\cite{antoniadis} for their problem deadline-scheduling-on-intervals, which only defines upper bounds.
\begin{restatable}{lemma}{lemmacut}\label{lemma:cut}
  For every $\alpha$-$\omega$ cut $(S, \bar S)$ in the network given in Figure~\ref{fig:flow} we have at least one of the following two lower bounds for the capacity $c(S)$ of the cut:
  $c(S) \geq P - \opdef(Q(S))$ or $c(S) \geq P - \exc(Q(\bar S))$, where $Q(S) \coloneqq \{ t \mid v_t \in S \}$.
\end{restatable}

\begin{restatable}{lemma}{lemmafeasibility}\label{lemma:feasibility}
  An instance of deadline-scheduling-with-processor-bounds is feasible if and only if $\opdef(Q) \leq 0$ and $\exc(Q) \leq 0$ for every $Q \subseteq T$.
\end{restatable}

\begin{definition}
  A time slot $t \in T$ is called \emph{engagement of processor $k$} if $t = \min B$ for some busy interval $B$ on processor $k$.
  A time slot $t \in T$ is just called \emph{engagement} if it is an engagement of processor $k$ for some $k \in [m]$.
\end{definition}
\begin{lemma}\label{lemma:critical}
  Let $Q \subseteq T$ be a set of time slots and $t \in T$ an engagement of processor $k \in [m]$.
  We call $Q$ a \emph{tight set for engagement $t$ of processor $k$} if $t \in Q$ and
  \begin{align}
    \fv(Q) &= \vol(Q) \text{,}\\
    \vol(t') &\geq k-1 &~\text{for all}~t' \in Q~\text{, and}\\
    \vol(t') &\geq k &~\text{for all}~t' \in Q~\text{with}~t' \geq t \text{.}
  \end{align}
  For every engagement $t$ of some processor $k \in [m]$ in the schedule $S_{\pltr}$ constructed by $\PLTR$, there exists a tight set $Q_t \subseteq T$ for engagement $t$ of processor $k$.
\end{lemma}
\begin{proof}
  Suppose for contradiction that there is some engagement $t \in T$ of processor $k \in [m]$ and no such $Q$ exists for $t$.
  We show that $\PLTR$ would have extended the idle interval on processor $k$ which ends at $t$.
  Consider the step in $\PLTR$ when $t$ was the result of $\keepidle$ on processor $k$.
  Let $l_{t'}$, $m_{t'}$ be the lower and upper bounds for $t' \in T$ right after the calculation of $t$ and the corresponding update of the bounds by $\keepidle$.
  We modify the bounds by decreasing $m_t$ by $1$.
  Note that at this point $m_{t'} \geq k$ for every $t' > t$ and $m_{t'} \geq k - 1$ for every $t'$.

  Consider $Q \subseteq T$ such that $t \in Q$ and $\fv(Q) < \vol(Q)$.
  Before our decrement of $m_t$ we had $m_Q \coloneqq \sum_{t' \in Q} m_{t'} \geq \vol(Q) > \fv(Q)$.
  The inequality $m_Q \geq \vol(Q)$ here follows since the upper bounds $m_{t'}$ are monotonically decreasing during $\PLTR$.
  Since our modification decreases $m_Q$ by at most $1$, we hence still have $m_Q \geq \fv(Q)$ after the decrement of $m_t$.
  Consider $Q \subseteq T$ such that $t \in Q$ and $\vol(t') < k - 1$ for some $t'$.
  At the step in $\PLTR$ considered by us, we hence have $m_{t'} \geq k - 1 > \vol(t')$.
  Before our decrement of $m_t$ we therefore have $m_Q > \vol(Q) \geq \fv(Q)$, which implies $m_Q \geq \fv(Q)$ after the decrement.
  Finally, consider $Q \subseteq T$ such that $t \in Q$ and $\vol(t') < k$ for some $t' > t$.
  At the step in $\PLTR$ considered by us, we again have $m_{t'} \geq k > \vol(t')$, which implies $m_Q \geq \fv(Q)$ after our decrement of $m_t$.
  In summary, if for $t$ no $Q$ exists as characterized in the proposition, the engagement of processor $k$ at $t$ could not have been the result of $\keepidle$ on processor $k$.
\end{proof}

\begin{lemma}
  We call a set $C_k \subseteq T$ \emph{critical set for processor $k$} if $C_k$ fulfills that
  \begin{itemize}
    \item
      $C_k \supseteq C_{k'}$ for every critical set for processor $k' > k$,
    \item
      $t \in C_k$ for every engagement $t$ of processor $k$,
    \item
      $\fv(C_k) = \vol(C_k)$,
    \item
      $\vol(t) \geq k - 1$ for every $t \in C_k$, and
    \item
      $\phi(C_k)$ is maximal.
  \end{itemize}
  For every processor $k \in [m]$ of $S_{\pltr}$ which is not completely idle, there exists a critical set $C_k$ for processor $k$.
\end{lemma}
\begin{proof}
  We show the existence by induction over the processors $m, \ldots, 1$.
  For processor $m$, consider the union of all tight sets over engagements of processor $m$.
  This set fulfills all conditions necessary except for the maximality in regard to $\phi$.
  Suppose that the critical sets $C_m, \ldots, C_{k+1}$ exist.
  Take $Q_k \subseteq T$ as the union of $C_{k+1}$ and all tight sets over engagements of processor $k$.
  By definition of $C_{k+1}$, we have $Q_k \supseteq C_{k'}$ for all $k' > k$.
  By construction of $Q_k$, every engagement $t$ of processor $k$ is contained in $Q_k$.
  Finally, we have $\fv(Q_k) = \vol(Q_k)$ and $\vol(t) \geq k-1$ for every $t \in Q_k$ since all sets in the union fulfill these properties.
\end{proof}

\subsection{Definitions Based on Critical Sets}

\begin{definition}
  For the critical set $C_k$ of some processor $k \in [m]$, we define $\crit(C_k) \coloneqq k$.
  Let $\succeq$ be the total order on the set of critical sets $C$ across all processors which corresponds to $\crit$, i.e.\ $C \succeq C'$ if and only if $\crit(C) \geq \crit(C')$.
  Equality in regard to $\succeq$ is denoted with $\sim$.
  We extend the definition of $\crit$ to general time slots $t \in T$ with $\crit(t) \coloneqq
      \max \{\crit(C) \mid C~\text{is critical set},  t \in C \}$ if $t \in C$ for some critical set $C$ and otherwise $\crit(t) \coloneqq 0$.
  We further extend $\crit$ to intervals $D \subseteq T$ with $\crit(D) \coloneqq \max \{ \crit(t) \mid t \in D \}$
\end{definition}

\begin{definition}
A nonempty interval $V \subseteq T$ is a \emph{valley} if $V$ is inclusion maximal on condition that $C \sim V$ for some fixed critical set $C$.
Let $D_1, \ldots, D_l$ be the maximal intervals contained in a critical set $C$.
A nonempty interval $V$ is a \emph{valley of $C$} if $V$ is exactly the valley between $D_{a}$ and $D_{a+1}$ for some $a < l$, i.e. $V = [\max D_a + 1, \min D_{a+1} - 1]$.
By choice of $C$ as critical set (property 1), a valley of $C$ is indeed a valley.
We define the jobs $J(V) \subseteq J$ for a valley $V$ as all jobs which are scheduled by $S_{\pltr}$ in every $t \in V$.
\end{definition}
\begin{definition}
  For a critical set $C$, an interval $D \subseteq T$ \emph{is a section of} $C$ if $D \cap C$ contains only full subintervals of $C$ and at least one subinterval of $C$.
  For a critical set $C$ and a section $D$ of $C$, the \emph{left valley} $V_l$ is the valley of $C$ ending at $\min (C \cap D) - 1$, if such a valley of $C$ exists.
  Symmetrically, the \emph{right valley} $V_r$ is the valley of $C$ starting at $\max (C \cap D) + 1$, if such a valley of $C$ exists.
\end{definition}

\begin{lemma}\label{lemma:valley}
  For every critical set $C$, every section $D \subseteq T$ of $C$, we have: if $\phi(C \cap D) \leq \crit(C) - \delta$ for some $\delta \in \mathbb{N}$, then the left valley $V_l$ or the right valley $V_r$ of $C$ and $D$ is defined and $|J(V_l)| + |J(V_r)| \geq \delta$.
  We take $|J(V)| \coloneqq 0$ if $V$ is not defined.
\end{lemma}
\begin{figure}
  \centering
  \includegraphics[width=1\textwidth]{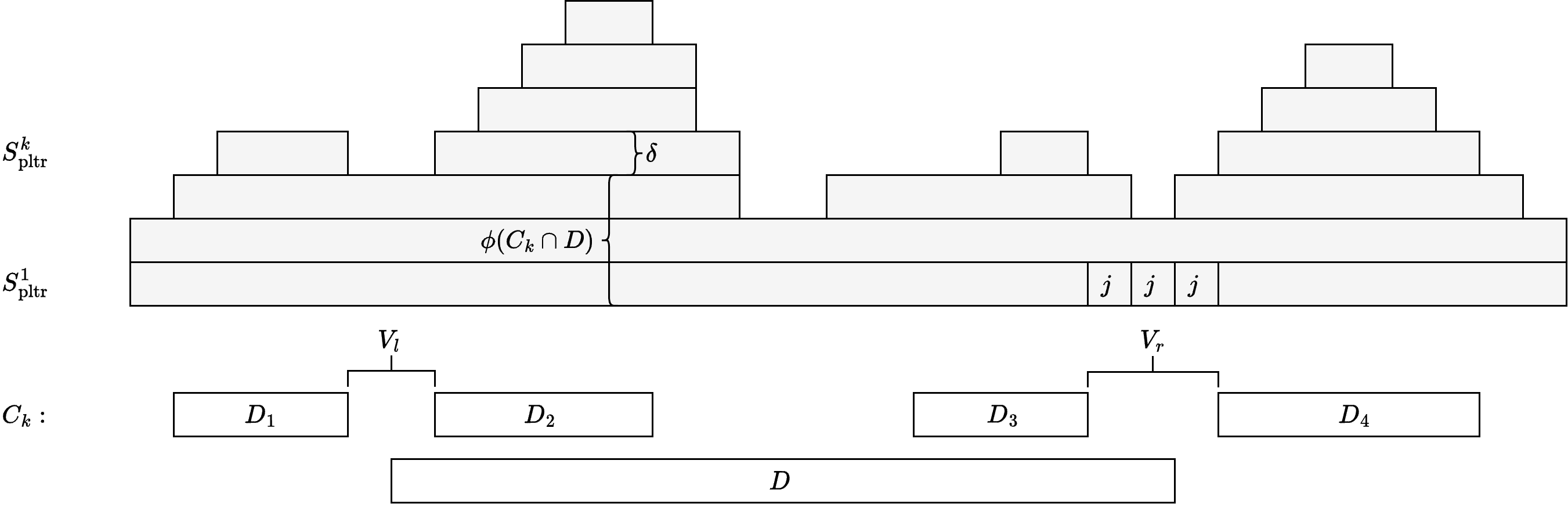}
  \caption{The left and right valley $V_l$ and $V_r$ of the critical set $C_k$ for processor $k$ and a section $D$ of $C_k$.
  Lemma~\ref{lemma:valley} guarantees that $\delta$ jobs are scheduled at every slot of $V_l$ or $V_r$.}\label{fig:left_right_valleys}
\end{figure}
\begin{proof}
  Refer to Figure~\ref{fig:left_right_valleys} for a visual sketch of the proposition.
  By choice of $C$ as critical set with $c \coloneqq \crit(C)$, we have $\vol(C \cap D) \geq (c-1) \cdot |C \cap D|$.
  If this inequality is fulfilled strictly, then with the premise
  $\fv(C \cap D) / |C \cap D| \leq c - \delta$ we directly get $\uv(C \cap D) / |C \cap D| > \delta - 1$.
  This implies that there are at least $\delta$ jobs $j$ scheduled in $C \cap D$ with $\uv(j, C \cap D) > 0$.
  Such jobs can be scheduled in the part of $C$ not contained in $D$, i.e.\ we must have $E_j \cap (C \setminus D) \neq \emptyset$ and hence the left valley $V_l$ or the right valley $V_r$ of $C$ and $D$ must be defined.
  Since these jobs $j$ are scheduled in $C$ only for the minimum amount possible, i.e.\ $\vol(j, C) = \fv(j, C) > 0$, they must be scheduled in every $t \in E_j \setminus C$ and are therefore contained in $J(V_l)$ or $J(V_r)$.

  If on the other hand we have equality, i.e.\ $\vol(C \cap D) = (c-1) \cdot |C \cap D|$, then let $t$ be an engagement of processor $c$.
  Since $\vol(t) > c-1$, we must have $t \notin C \cap D$.
  By the same argument as before, we have that if $\fv(C \cap D) / |C \cap D| \leq c - \delta$, then $\uv(C \cap D) / |C \cap D| \geq \delta - 1$.
  Let $J' \coloneqq \{ j \in J \mid \uv(j, C \cap D) > 0 \}$.
  Since $\uv(j, C \cap D) \leq | C \cap D |$ for every $j \in J$, we have $|J'| \geq \delta - 1$.
  If this lower bound is fulfilled with equality, then every $j \in J'$ must be scheduled in every time slot of $C \cap D$ and hence $\fv(J', C \setminus D) = \vol(J', C \setminus D)$.
  Now suppose for contradiction that all jobs $j$ scheduled during $C \setminus D$ which are not contained in $J'$ have $E_j \cap C \cap D = \emptyset$.
  Then $\fv(C \setminus D) = \vol(C \setminus D)$ and
  we get $\phi(C \setminus D) > \phi(C)$ since by case assumption $\vol(C \cap D) / |C \cap D| = (c-1) < \phi(C)$.
  With $\vol(t) \leq c - 1$ for every $t \in C \cap D$, we know that $\crit(C \cap D) \leq k$ and therefore $C \setminus D$ is still a critical set for processor $c$ but has higher density than $C$, contradicting the choice of $C$.
  Therefore, there must exist a job $j \notin J'$ scheduled in $C \setminus D$ with an execution interval intersecting $C \cap D$.
  In any case, we have at least $\delta$ jobs scheduled in $C$ with an execution interval intersecting both $C \setminus D$ and $C \cap D$.
  This implies that the left valley $V_l$ or the right valley $V_r$ of $C$ and $D$ exists and that at least $\delta$ jobs are contained in $J(V_l)$ or $J(V_r)$.
\end{proof}

\section{Modification of the PLTR-Schedule for Analysis}\label{section:approximation}
In this section we modify the schedule $S_{\pltr}$ returned by $\PLTR$ in two steps.
We stress that this is for the analysis only and not part of $\PLTR$.
The first step augments specific processors with auxiliary busy slots such that in every critical set $C$ at least the first $\crit(C)$ processors are busy all the time.
For the single processor $\LTR$ algorithm, the crucial property for the approximation guarantee is that every idle interval of $S_{\opt}$ can intersect at most $2$ distinct idle intervals of the schedule returned by $\LTR$.
The second modification step of $S_{\pltr}$ is more involved and establishes this crucial property on every processor $k \in [m]$ by making use of Lemma~\ref{lemma:valley}.
More specifically, it will establish the stronger property that $\hat \phi(B) > k - 1$ for every $\busy$ interval $B$ on processor $k$ with $\crit(B) \geq 2$, i.e.\ that every feasible schedule requires $k$ busy processors at some point during $B$.
Idle intervals surrounded by only busy intervals $B$ with $\crit(B) \leq 1$ are then handled in Lemma~\ref{lemma:constellation} with essentially the same argument as for the single processor $\LTR$ algorithm.
By making sure that the modifications cannot decrease the costs of our schedule, we obtain an upper bound for the costs of $S_{\pltr}$.

\subsection{Augmentation and Realignment}
We transform $S_{\pltr}$ into the \emph{augmented schedule $S_{\aug}$} by adding for every $t$ with $k \coloneqq \crit(t) \geq 2$ and $\vol(t) = k-1$ an auxiliary busy slot on processor $k$.
No job is scheduled in this auxiliary busy slot on processor $k$ and it does also not count towards the volume of this slot.
It merely forces processor $k$ to be in the on-state at time $k$ while allowing us to keep thinking in terms of $\idle$ and $\busy$ intervals in our analysis of the costs.

\begin{lemma}\label{lemma:augmented}
  In $S_{\aug}$ processors $1, \ldots, \crit(t)$ are busy in every slot $t \in T$ with $\crit(t) \geq 2$.
\end{lemma}
\begin{proof}
  The property directly follows from our choice of the critical sets, the definition of $\crit(t)$, and the construction of $S_{\aug}$.
\end{proof}

As a next step, we transform $S_{\aug}$ into the \emph{realigned schedule $S_{\real}$} using Algorithm~\ref{alg:real}. We briefly sketch the ideas behind this realignment.
Lemma~\ref{lemma:augmented} guarantees us that every busy interval $B$ on processor $k$ is a section of the critical set $C$ with $C \sim B$.
It also guarantees that the left and right valley $V_l, V_r$ of $C$ and $B$ do not end within an idle interval on processor $k$.
Lemma~\ref{lemma:valley} in turn implies that if the density of $B$ is too small to guarantee that $S_{\opt}$ has to use processor $k$ during $B$, i.e.\ if $\hat \phi(B) \leq k - 1$, then $V_l$ or $V_r$ is defined and there is some $j$ scheduled in every slot of $V_l$ or $V_r$.
Let $V$ be the corresponding left or right valley of $C$ and $D$ for which such a job $j$ exists.
Instead of scheduling $j$ on the processors below $k$, we can schedule $j$ on processor $k$ in idle time slots during $V$.
This merges the busy interval $B$ with at least one neighbouring busy interval on processor $k$.
In the definition of the realignment, we will call this process of filling the idle slots during $V$ on processor $k$ \emph{closing of valley $V$ on processor $k$}.
The corresponding subroutine is called $\close(k, V)$.

The crucial part is ensuring that this merging of busy intervals by closing a valley continues to be possible throughout the realignment whenever we encounter a busy interval with a density too small.
For this purpose, we go through the busy intervals on each processor in decreasing order of their criticality, i.e.\ in the order of $\succeq$.
We also allow every busy slot to be used twice for the realignment by introducing further auxiliary busy slots, since for a section $D$ of the critical set $C$, both the right and the left valley might be closed on processor $k$ in the worst case.
This allows us to maintain the invariants stated in Lemma~\ref{lemma:invariant} during the realignment process, which correspond to the initial properties of Lemma~\ref{lemma:valley} and~\ref{lemma:augmented} for $S_{\aug}$.

\begin{algorithm}
  \caption{Realignment of $S_{\aug}$ for analysis only}\label{alg:real}
\begin{algorithmic}
  \State{$\res(V) \gets 2 |J(V)|$ for every valley $V$}
  \For{$k \gets m$ \textrm{to} $1$}
    \State{\Call{$\fillop$}{$k, T$}}
    \State{$\res(V) \gets \res(V) - 1$ for every valley $V$
    s.t.\ some $V'$ with $V' \cap V \neq \emptyset$ was closed on processor $k$}
  \EndFor{}
\algrule
  \Function{$\fillop$}{$k, V$}
    \If{$\crit(V) \leq 1$}
      \State{\Return{}}
    \EndIf{}
    \State{let $C$ be the critical set s.t.\ $C \sim V$}
    \While{exists busy interval $B \subseteq V$ on processor $k$
    with $B \sim V$ and $\hat \phi(B) \leq k - 1$}
      \State{let $V_l, V_r$ be the left and right valley for $C$ and interval $B$ (if $B$ is a section of $C$)}
      \If{$V_l$ exists and $\res(V_l) > 0$}
        \State{\Call{$\close$}{$k, V_l$}}
      \ElsIf{$V_r$ exists and $\res(V_r) > 0$}
        \State{\Call{$\close$}{$k, V_r$}}
      \EndIf{}
    \EndWhile{}
    \For{every valley $V' \subseteq V$ of $C$ which has not been closed on $k$}
      \State{\Call{$\fillop$}{$k, V'$}}
    \EndFor{}
  \EndFunction{}

  \Function{$\close$}{$k, V$}
    \For{every $t \in V$ which is idle on processor $k$}
      \If{processors $1, \ldots, k-1$ are idle at $t$}
        \State{introduce new auxiliary busy slot on processor $k$ at time $t$}
      \Else{}
        \State{move busy slot at time $t$ of highest processor
        among $1, \ldots, k-1$ to processor $k$ at $t$}
      \EndIf{}
    \EndFor{}
  \EndFunction{}
\end{algorithmic}
\end{algorithm}

\subsection{Invariants for Realignment}
\begin{restatable}{lemma}{lemmainvariant}\label{lemma:invariant}
  For an arbitrary step during the realignment of $S_{\aug}$ and a valley $V \subseteq T$, let the \emph{critical processor $k_V$ for $V$} be the highest processor such that
  \begin{itemize}
    \item
      processor $k_V$ is not fully filled yet, i.e.\ $\fillop(k_V, T)$ has not yet returned,
    \item
      no $V' \supseteq V$ has been closed on $k_V$ so far, and
    \item
      there is a (full) busy interval $B \subseteq V$ on processor $k_V$.
  \end{itemize}
  We take $k_V \coloneqq 0$ if no such processor exists.
  At every step in the realignment of $S_{\aug}$ the following invariants hold for every valley $V$, where $C$ denotes the critical set with $C \sim V$.
  \begin{enumerate}
    \item
      If $\phi(C \cap D) \leq k_V - \delta$ for some $\delta \in \mathbb{N}$, some section $D \subseteq V$ of $C$, then the left valley $V_l$ or the right valley $V_r$ of $C, D$ exists and $\res(V_l) + \res(V_r) \geq 2 \delta$.
    \item
      For every $t \in C \cap V$, processors $1, \ldots, k_V$ are busy at $t$.
    \item
      Every busy interval $B \subseteq V$ on processor $k_V$ with $B \sim V$ is a section of $C$.
  \end{enumerate}
\end{restatable}

\begin{restatable}{lemma}{lemmasreal}\label{lemma:s_real}
  The resulting schedule $S_{\real}$ of the realignment of $S_{\aug}$ is defined.
\end{restatable}

\begin{lemma}\label{lemma:phi_prop}
  For every processor $k \in [m]$ and every busy interval $B$ on processor $k$ in $S_{\real}$ with $\crit(B) \geq 2$, we have $\hat \phi(B) > k - 1$.
\end{lemma}
\begin{proof}
  We show that $\fillop(k, T)$ establishes the property on processor $k$.
  The claim then follows since $\fillop(k, T)$ does not change the schedules of processors above $k$.
  We know that on processor $k$ busy intervals are only extended, since in $\fillop(k, T)$ we only close valleys for busy intervals $B$ on $k$ which are a section of the corresponding critical set $C$.
  Let $B \subseteq V$ be a busy interval on processor $k$ in $S_{\real}$ with $B \sim V$ and $\crit(B) \geq 2$.
  No valley $W \supseteq V$ can have been closed on $k$ since otherwise there would be no $B \subseteq V$ in $S_{\real}$.
  Therefore, at some point $\fillop(k, V)$ must be called.
  Consider the point in $\fillop(k, V)$ when the while-loop terminates.
  Clearly at this point all busy intervals $B' \subseteq V$ with $B' \sim V$ on processor $k$ have $\hat \phi(B') > k - 1$.
  At this point there must also be at least one such $B'$ for $B$ to be a busy interval on $k$ in $S_{\real}$ with $B \sim V$ and $B \subseteq V$.
  In particular, one such $B'$ must have $B' \subseteq B$, which directly implies $\hat \phi(B) \geq \hat \phi(B') > k - 1$.
\end{proof}
While with Lemma~\ref{lemma:phi_prop} we have our desired property for busy intervals $B$ of $\crit(B) \geq 2$, we still have to handle busy intervals of $\crit(B) \leq 1$.
To be precise, we have to handle idle intervals which are surrounded only by busy intervals $B$ of $\crit(B) \leq 1$.
We will show that this constellation can only occur in $S_{\real}$ on processor $1$ and that the realignment has not done any modifications in these intervals, i.e.\ $S_{\pltr}$ and $S_{\real}$ do not differ for these intervals.
With the same argument as for the original single-processor Left-to-Right algorithm, we then get that at least one processor has to be busy in any schedule during these intervals.

\begin{lemma}\label{lemma:engagements}
  The realignment of $S_{\aug}$ does not create new engagement times but may only change the corresponding processor being engaged, i.e.\ if $t \in T$ is an engagement of some processor $k$ in $S_{\real}$, then $t$ is also an engagement of some processor $k'$ in $S_{\aug}$.
\end{lemma}
\begin{proof}
  Consider the first step in the realignment of $S_{\aug}$ in which some $t \in T$ becomes an engagement of some processor $k'$ where $t$ was no engagement of any processor before this step.
  This step must be the closing of some valley $V$ on some processor $k > k'$:
  On processor $k$, we have seen that closing some valley can only merge busy intervals.
  On processors above $k$, the schedule does not change.
  Busy slots on processors $k'' < k$ are only removed (by definition of $\close$), therefore $t-1$ must have been busy on processor $k'$ and idle on $k' + 1, \ldots, k$ before the close.
  \begin{figure}[H]
    \centering
		\includegraphics[scale=0.4]{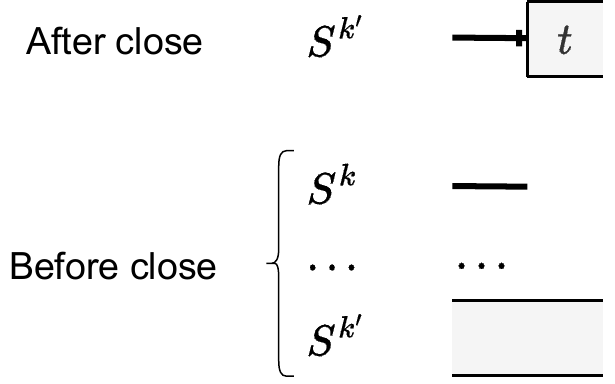}
  \end{figure}
  If $t \in V$, then processor $k' + 1$ (or $k$) must have been busy before at $t$.
  Hence $t$ was already an engagement before the close, contradicting our initial choice of $t$.
  If $t \notin V$, then $t \succ V$.
  Let $W$ be the valley such that $V$ is closed during $\fillop(k, W)$, hence $W \supset V$.
  If $t \in W$, then $t \sim C_W$ and $t \in C_W$.
  By Invariant 2, processors $1, \ldots, k_W = k$ are busy at $t$ before the close.
  Again, this implies that $t$ was an engagement before the close already, contradicting our choice of $t$.
  If $t \notin W$, then let $W'$ be the valley with $t \sim W$ and $t \in W'$.
  We have $W \prec t \sim W'$ and $W' \supset W$ and $t \in C_{W'}$.
  Therefore $k_{W'} \geq k_W = k$ and Invariant 2 implies that processors $1, \ldots, k$ are busy at $t$ before the close.
  Hence, $t$ was an engagement before the close already, again contradicting our initial choice of $t$.
\end{proof}
\begin{lemma}\label{lemma:constellation}
  Let $I$ be an idle interval in $S_{\real}$ on some processor $k$ and let $B_l, B_r$ be the busy intervals on $k$ directly to the left and right of $I$ with $\crit(B_l) \leq 1$ and $\crit(B_r) \leq 1$.
  Allow $B_l$ to be empty, i.e.\ we might have $\min I = 0$, but $B_r$ must be nonempty, i.e.\ $\max I < d$.
  Then we must have $k = 1$ and $\hat \phi(B_l \cup I \cup B_r) > 0$.
\end{lemma}
\begin{proof}
  By Lemma~\ref{lemma:engagements} and $\crit(B_r) \leq 1$, we know that $\min B_r$ is an engagement of processor $1$ in $S_{\aug}$.
  Hence $\max I$ is idle in $S_{\aug}$ on processor $1$ and hence idle on all processors (by stair-property in $S_{\aug}$).
  Since no jobs are scheduled at $\max I$, we know that $\crit(\max I) \leq 1$ and $J(V) = \emptyset$ for all valleys $V$ containing the slot $\max I$, and hence also $\res(V) = 0$ at all times during the realignment.
  Therefore, no $V$ intersecting $[\max I, \max B_r]$ was closed during the realignment on any processor since this $V$ would contain $\max I$.
  Since $B_r$ is a busy interval with $\crit(B_r) \leq 1$ (i.e.\ not containing engagements of processors above $1$ in $S_{\aug}$), we must then have $k = 1$.
  For $I$ to be idle on processor $k = 1$ in $S_{\real}$ and $\crit(I) \geq 2$, some $V \succeq I$ with $V \cap I \neq \emptyset$ and hence $V \supseteq \max I$ would have to been closed, which contradicts what we have just shown.
  Therefore $\crit(I) \leq 1$ and no valley $V$ with $V \cap [\min B_l - 1, \max B_r + 1] \neq \emptyset$ can have been closed during the realignment.
  Therefore, the constellation occurs exactly in the same way in $S_{\aug}$ and $S_{\pltr}$ on processor $1$.
  In other words, for processor $1$ in both $S_{\aug}$ and $S_{\pltr}$, $B_l$ and $B_r$ are busy intervals and $I$ is an idle interval.

  Let $j$ be the single job scheduled at time slot $\min B_r$.
  We conclude by showing that $E_j \subseteq I \cup B_r$ and therefore $\hat \phi(I \cup B_r) > 0$.
  Otherwise, $j$ could be scheduled at $\min I$ or $\max B_r + 1$.
  In the first case, $\PLTR$ would have extended $B_l$ by scheduling $j$ at time $\min I$ instead of at $\min B_r$.
  In the second case, $\PLTR$ would have extended the idle interval $I$ by scheduling $j$ at $\max B_r + 1$ instead of at $\min B_r$.
\end{proof}

\begin{lemma}\label{lemma:intersection}
  For every processor $k$, every idle interval on processor $k$ in $S_{\opt}$ intersects at most two distinct idle intervals of processor $k$ in $S_{\real}$.
\end{lemma}
\begin{proof}
  Let $I_{\opt}$ be an idle interval in $S_{\opt}$ on processor $k$ intersecting three distinct idle intervals of processor $k$ in $S_{\real}$.
  Let $I$ be the middle one of these three idle intervals.
  Lemma~\ref{lemma:constellation} and Lemma~\ref{lemma:phi_prop} imply that $k$ busy processors are required during $I$ and its neighboring busy intervals.
  This makes it impossible for $S_{\opt}$ to be idle on processor $k$ during the whole interval $I_{\opt}$.
\end{proof}

\subsection{Approximation Guarantee}
Lemma~\ref{lemma:intersection} finally allows us to bound the costs of the schedule $S_{\real}$ with the same arguments as in the proof for the single-processor LTR algorithm of~\cite{irani_left_to_right_soda_2003}.
We complement this with an argument that the augmentation and realignment could have only increased the costs of $S_{\pltr}$ and that we have hence also bounded the costs of the schedule returned by our algorithm $\PLTR$.

\begin{theorem}\label{theorem:approximation}
  Algorithm $\PLTR$ constructs a schedule of costs at most $2 \OPT + P$.
\end{theorem}
\begin{proof}
  We begin by bounding $\costs(S_{\real})$ as in the proposition.
  First, we show that $\idle(S^k_{\real}) \leq 2 \off(S^k_{\opt}) + \on(S^k_{\opt})$ for every processor $k \in [m]$.
  Let $\mathcal{I}_1$ be the set of idle intervals on $S^k_{\real}$ which intersect some $\off$ interval of $S^k_{\opt}$.
  Lemma~\ref{lemma:intersection} implies that $\mathcal{I}_1$ contains as most twice as many intervals as there are $\off$ intervals in $S^k_{\opt}$.
  Since the costs of each idle interval are at most $q$, and the costs of each off interval are exactly $q$, the costs of all idle intervals in $\mathcal{I}_1$ is bounded by $2 \off(S^k_{\opt})$.
  Let $\mathcal{I}_2$ be the set of idle intervals on $S^k_{\real}$ which do not intersect any $\off$ interval in $S^k_{\opt}$.
  The total length of these intervals is naturally bounded by $\on(S^k_{\opt})$.

  We continue by showing that $\busy(S_{\real}) \leq 2 P$.
  By construction of $S_{\aug}$ and the definition of $\res$ and $\close$, we introduce at most as many auxiliary busy slots at every slot $t \in T$ as there are jobs scheduled at $t$ in $S_{\pltr}$.
  For $S_{\aug}$, an auxiliary busy slot is only added for $t$ with $\crit(t) \geq 2$ and hence $\vol(t) \geq 1$.
  Furthermore, initially $\res(V) = 2 |J(V)|$ for every valley $V$ and $\res(V)$ is decremented if some $V'$ intersecting $V$ is closed during $\fillop(k, T)$.
  During $\fillop(k, T)$ at most a single $V'$ containing $t$ is closed for every $t \in T$.
  Finally, auxiliary busy slots introduced by $S_{\aug}$ are used in the subroutine $\close$.
  This establishes the lower bound $\costs(S_{\real}) = \idle(S_{\real}) + \busy(S_{\real}) \leq 2 \off(S_{\opt}) + \on(S_{\opt}) + 2 P \leq 2 \OPT + P$ for our realigned schedule.

  We complete the proof by arguing that $\costs(S_{\pltr}) \leq \costs(S_{\real})$ since transforming $S_{\real}$ back into $S_{\pltr}$ does not increase the costs of the schedule.
  Removing the auxiliary busy slots clearly cannot increase the costs.
  Since the realignment of $S_{\aug}$ only moves busy slots between processors, but not between different time slots, we can easily restore $S_{\pltr}$ (up to permutations of the jobs scheduled on the busy processors at the same time slot) by moving all busy slots back down to the lower numbered processors.
  By the same argument as in Lemma~\ref{lemma:stair_property_opt}, this does not increase the total costs of the schedule.
\end{proof}

\section{Running Time}\label{section:running_time}
\begin{theorem}\label{theorem:running_time}
  Algorithm $\PLTR$ has a running time of $\mathcal{O}(n f \log d)$ where $f$ denotes the time needed for finding a maximum flow in a network with $\mathcal{O}(n)$ nodes.
\end{theorem}
\begin{proof}
  First observe that every busy interval is created by a pair of calls to $\keepidle$ and $\keepbusy$, respectively.
  We begin by bounding the number of busy intervals across all processors in $S_{\pltr}$ by $n$.
  Note that if $\keepidle$ returns $d$, then we do not have to calculate $\keepbusy$ from $d$ on.
  Therefore, the total number of calls to $\keepidle$ and $\keepbusy$ is then bounded by $n + m$.
  If $m > n$ we can restrict our algorithm to use the first $n$ processors only, as there cannot be more than $n$ processors scheduling jobs at the same time.
  We derive the upper bound of $n$ for the number of busy intervals across all processors by constructing an injective mapping $g$ from the set of busy intervals to the jobs $J$.
  For this construction of $g$ we consider the busy intervals in the same order as the algorithm, i.e.\ from Top-Left to Bottom-Right.
  We construct $g$ such that $g(B) = j$ only if $d_j \in B$.

  Suppose we have constructed such a mapping for busy intervals on processors $m, \ldots, k$ up to some busy interval $B$ on $k$.
  We call a busy interval $B'$ in $S_{\pltr}$ on processor $l \in [m]$ a \emph{plateau on processor $l$}, if all slots of $B'$ are idle for all processors above $l$.
  Observe that plateaus (even across different processors) cannot intersect, which implies an ordering of the plateaus from left to right.
  Let $B'$ be the last plateau with $B' \subseteq B$ and let $l \geq k$ be the processor for which this busy interval $B'$ is a plateau.
  By construction of $g$ and the choice of $B'$, there are at most $l - k$ distinct jobs $j$ with $d_j \in [\min B', \max B]$ already mapped to by $g$.
  This is since at most $l - k$ busy intervals on processors $k+1, \ldots, m$ intersect the interval $[\min B', \max B]$.
  Let $Q_t \subseteq T$ be a tight set over engagement $t \coloneqq \min B$ of processor $l$.
  Let $J' \coloneqq \{ j_1, \ldots, j_{l} \}$ be the $l$ distinct jobs scheduled at $t$.
  We know that $\max B + 1 \notin C_t$ since $\vol(\max B + 1) < k \leq l$ and $\max B + 1 > t$.
  With $\vol(j, Q_t) = \fv(j, Q_t) > 0$ for every $j \in J'$, we know that every job $j \in J'$ with $d_j > \max B$ is scheduled at slot $\max B + 1$.
  Hence there are at least $l - (k-1)$ distinct jobs $j \in J'$ with $d_j \in [\min B', \max B]$ and there must be at least one such job $j^*$ which is not mapped to by $g$ so far and which we therefore can assign to $B$.

  Having bounded the number of calls to $\keepidle$ and $\keepbusy$ by $\mathcal{O}(n)$, the final step is to bound the running time of these two subroutines by $\mathcal{O}(f \log d)$.
  A slight modification to the flow-network of Figure~\ref{fig:flow} suffices to have only $\mathcal{O}(n)$ nodes.
  The idea here is to partition the time horizon $T$ into $\mathcal{O}(n)$ \textit{time intervals} instead of $\mathcal{O}(d)$ individual time slots.
  Since this is a standard problem as laid out in e.g.\ Chapter 5 of~\cite{brucker_scheduling_algorithms}, we only sketch the main points relevant to our setting in the following.

  The partition of the time horizon into time intervals is done by using the release times and deadlines as splitting points of the time horizon and scaling the capacities of the incoming and outgoing edges by the length of the time interval.
  For our generalization we additionally have to split whenever an upper or lower bound $l_t, m_t$ changes.
  Since we have already bounded the number of such times by $2n$ in the first part of this proof, there are only $\mathcal{O}(n)$ time intervals and hence also $\mathcal{O}(n)$ nodes in the flow network.
  Also note that constructing the sub-schedules within the time intervals is a much simpler scheduling problem, since by construction, for every time interval and every job $j$, the execution interval $E_j$ either completely contains the time interval or does not intersect it.
  Such a sub-schedule can be computed in $O(n^2)$, as laid out in Chapter 5 of~\cite{brucker_scheduling_algorithms}.
  With the feasibility check running in time $\mathcal{O}(f)$, each call to $\keepidle$ and $\keepbusy$ can be completed in $\mathcal{O}(f \log d)$ using binary search on the remaining time horizon.
\end{proof}

\section*{Acknowledgement}
Thanks to Prof.\ Dr.\ Susanne Albers for her supervision during my studies.
The idea of generalizing the Left-to-Right algorithm emerged in discussions during this supervision.
This work was supported by the Research Training Network of the Deutsche Forschungsgemeinschaft (DFG) (378803395: ConVeY).
\bibliographystyle{abbrvnat}
\bibliography{references}

\begin{thebibliography}{8}
\providecommand{\natexlab}[1]{#1}
\providecommand{\url}[1]{\texttt{#1}}
\expandafter\ifx\csname urlstyle\endcsname\relax
  \providecommand{\doi}[1]{doi: #1}\else
  \providecommand{\doi}{doi: \begingroup \urlstyle{rm}\Url}\fi

\bibitem[Antoniadis et~al.(2020)Antoniadis, Garg, Kumar, and Kumar]{antoniadis}
A.~Antoniadis, N.~Garg, G.~Kumar, and N.~Kumar.
\newblock Parallel machine scheduling to minimize energy consumption.
\newblock In \emph{Proceedings of the Thirty-First Annual ACM-SIAM Symposium on
  Discrete Algorithms}, SODA '20, page 2758–2769, USA, 2020. Society for
  Industrial and Applied Mathematics.

\bibitem[Antoniadis et~al.(2021)Antoniadis, Kumar, and Kumar]{skeletons}
A.~Antoniadis, G.~Kumar, and N.~Kumar.
\newblock {Skeletons and Minimum Energy Scheduling}.
\newblock In H.-K. Ahn and K.~Sadakane, editors, \emph{32nd International
  Symposium on Algorithms and Computation (ISAAC 2021)}, volume 212 of
  \emph{Leibniz International Proceedings in Informatics (LIPIcs)}, pages
  51:1--51:16, Dagstuhl, Germany, 2021. Schloss Dagstuhl -- Leibniz-Zentrum
  f{\"u}r Informatik.
\newblock ISBN 978-3-95977-214-3.
\newblock \doi{10.4230/LIPIcs.ISAAC.2021.51}.
\newblock URL \url{https://drops.dagstuhl.de/opus/volltexte/2021/15484}.

\bibitem[Baptiste(2006)]{baptiste_unit_jobs}
P.~Baptiste.
\newblock Scheduling unit tasks to minimize the number of idle periods: A
  polynomial time algorithm for offline dynamic power management.
\newblock In \emph{Proceedings of the Seventeenth Annual ACM-SIAM Symposium on
  Discrete Algorithm}, SODA '06, page 364–367, USA, 2006. Society for
  Industrial and Applied Mathematics.
\newblock ISBN 0898716055.

\bibitem[Baptiste et~al.(2007)Baptiste, Chrobak, and
  D{\"u}rr]{baptiste_general_jobs}
P.~Baptiste, M.~Chrobak, and C.~D{\"u}rr.
\newblock Polynomial time algorithms for minimum energy scheduling.
\newblock In \emph{European Symposium on Algorithms}, pages 136--150. Springer,
  2007.

\bibitem[Brucker(2004)]{brucker_scheduling_algorithms}
P.~Brucker.
\newblock \emph{Scheduling Algorithms}, volume~47.
\newblock 01 2004.
\newblock \doi{10.2307/3010416}.

\bibitem[Demaine et~al.(2007)Demaine, Ghodsi, Hajiaghayi, Sayedi-Roshkhar, and
  Zadimoghaddam]{demaine}
E.~D. Demaine, M.~Ghodsi, M.~T. Hajiaghayi, A.~S. Sayedi-Roshkhar, and
  M.~Zadimoghaddam.
\newblock Scheduling to minimize gaps and power consumption.
\newblock In \emph{Proceedings of the Nineteenth Annual ACM Symposium on
  Parallel Algorithms and Architectures}, SPAA '07, page 46–54, New York, NY,
  USA, 2007. Association for Computing Machinery.
\newblock ISBN 9781595936677.
\newblock \doi{10.1145/1248377.1248385}.

\bibitem[Graham et~al.(1979)Graham, Lawler, Lenstra, and {Rinnooy Kan}]{graham}
R.~Graham, E.~Lawler, J.~Lenstra, and A.~{Rinnooy Kan}.
\newblock Optimization and approximation in deterministic sequencing and
  scheduling : a survey.
\newblock \emph{Annals of Discrete Mathematics}, 5:\penalty0 287--326, 1979.
\newblock ISSN 0167-5060.
\newblock \doi{10.1016/S0167-5060(08)70356-X}.

\bibitem[Irani et~al.(2003)Irani, Shukla, and
  Gupta]{irani_left_to_right_soda_2003}
S.~Irani, S.~K. Shukla, and R.~K. Gupta.
\newblock Algorithms for power savings.
\newblock In \emph{Proceedings of the Fourteenth Annual {ACM-SIAM} Symposium on
  Discrete Algorithms, January 12-14, 2003, Baltimore, Maryland, {USA}}, pages
  37--46. {ACM/SIAM}, 2003.
\newblock URL \url{http://dl.acm.org/citation.cfm?id=644108.644115}.

\end{thebibliography}
\section*{Appendix}
\lemmastairproperty*
\begin{proof}
  Let $S$ be an optimal schedule.
  We transform $S$ such that it fulfills the stair-property without increasing its costs and while maintaining feasibility.
  Let $k, k' \in [m]$ be two processors with $k' < k$ and job $j \in J$ scheduled on processor $k$ in time slot $t \in T$ while $k'$ is idle in $t$.
  Let $I$ be the idle interval on processor $k'$ containing $t$.
  We now move all jobs scheduled on processor $k$ during $I$ to be scheduled on processor $k'$ instead.
  Since $I$ is a maximal interval for which processor $k'$ is idle, this modification does not increase the combined costs of processors $k'$ and $k$.
  The modification also moves at least job $j$ from processor $k$ down to $k'$ while not moving any job from processor $k'$ to $k$.
  Jobs are only moved between processors at the same time slot and only to slots of processor $k'$ which are idle, hence the resulting schedule is still feasible.
  This modification can be repeated until the schedule has the desired property.
\end{proof}

\lemmaflowfeasibility*
\begin{proof}
  Let $f$ be an $\alpha$-$\omega$ flow of value $|f| = P$.
  We construct a feasible schedule from $f$ respecting the lower and upper bounds given by $l_t$ and $m_t$.
  For every $j \in J$ and  $t \in T$, if $f(u_j, v_t) = 1$, then schedule $j$ at slot $t$ on the lowest-numbered processor not scheduling some other job.
  Since $|f| = P$ and the capacity of the cut $c(\{\alpha\}, V \setminus \{\alpha\}) = P$, we have $f_{in}(u_j) = p_j$ for every $j \in J$.
  Hence $f_{out}(u_j) = \sum_{t \in E_j} f_{in}(v_t) = p_j$.
  Hence every job $j$ is scheduled in $p_j$ distinct time slots within its execution interval.

  The schedule respects the upper bounds $m_t$, since $c(v_t, \gamma) + c(v_t, \omega) \leq m_t - l_t + l_t$ and hence for every $t$ at most $m_t$ jobs are scheduled at $t$.
  The schedule respects the lower bounds $l_t$, since
  $c(V \setminus \{\omega\}, \{\omega\}) = P$ and hence
  $f(v_t, \omega) = l_t$ for every slot $t \in T$.
  By flow conservation we then have $f_{in}(v_t) \geq l_t$, which implies that at least $l_t$ jobs are scheduled at every slot $t$.

  For the other direction consider a feasible schedule respecting the lower and upper bounds $l_t, m_t$.
  We construct a flow $f$ of value $P$ and show that it is maximal.
  If $j$ is scheduled at slot $t$ and hence $t \in E_j$,
  define $f(u_j, v_t) = 1$, otherwise $f(u_j, v_t) = 0$.
  Define $f(\alpha, u_j) = p_j$ for every $j \in J$.
  Hence we have $f_{in}(u_j) = p_j$
  and $f_{out}(u_j)$ must be  $p_j$ since this corresponds to the number of distinct time slots in which $j$ is scheduled.
  Define $f(v_t, \omega) = l_t$ for every slot $t \in T$.
  Define $f(v_t, \gamma) = f_{in}(v_t) - l_t$.
  We have $f(v_t, \gamma) \leq m_t - l_t$ since $f_{in}(v_t)$ corresponds to the number $\vol(t)$ of jobs scheduled at $t$, which is at most $m_t$.
  We also have $f_{out}(v_t) = f_{in}(v_t) - l_t + l_t = f_{in}(v_t)$.
  Define $f(\gamma, \omega) = P - \sum_{t \in T} l_t$.
  Then $f_{in}(\gamma) = \sum_{t \in T} f_{in}(v_t) - l_t
  = \sum_{t \in T} \vol(t) - \sum_{t \in T} l_t$.
  Since the schedule is feasible, we have $\sum_{t \in T} \vol(t) = P$ and finally the flow conservation $f_{in}(\gamma) = P - \sum_{t \in T} l_t = f_{out}(\gamma)$.
\end{proof}

\lemmacut*
\begin{proof}
  Let $(S, \bar S)$ be an $\alpha$-$\omega$ cut and let $J(S) \coloneqq \{j \mid u_j \in S\}$.
  We consider the contribution of every node of $S$ to the capacity $c(S)$ of the cut.
  First consider the case that $\gamma \notin S$.
  \begin{itemize}
    \item Node $\alpha$: $\sum_{j \in J(\bar S)} p_j$
    \item Node $u_j$: $|\{v_t \in \bar S \mid t \in E_j\}| = | E_j \setminus Q(S) | \geq p_j - \fv(j, Q(S))$
    \item Node $v_t$: $l_t + m_t - l_t = m_t$
  \end{itemize}
  The inequality for node $u_j$ follows since $\fv(j, Q(S)) = \max \{0, p_j - |E_j \setminus Q(S)| \}$.
  In total, we can bound the capacity from below with
  \begin{align}
    c(S) &\geq \sum_{j \in J(\bar S)} p_j + \sum_{j \in J(S)} p_j - \fv(j, Q(S)) + \sum_{t \in Q(S)} m_t
    \\ &= P - \fv(J(S), Q(S)) + \sum_{t \in Q(S)} m_t
    \\ &\geq P - \opdef(Q(S))\text{.}
  \end{align}
  If $\gamma \in S$, we have the following contributions of nodes in $S$ to the capacity of the cut:
  \begin{itemize}
    \item Node $\alpha$: $\sum_{j \in J(\bar S)} p_j \geq \pv(J(\bar S), Q(\bar S))$
    \item Node $u_j$: $| E_j \setminus Q(S) |
      = | E_j \cap Q(\bar S)| \geq \pv(j, Q(\bar S))$
    \item Node $v_t$: $l_t$
    \item Node $\gamma$: $P - \sum_{t \in T} l_t$
  \end{itemize}
  In total, we obtain the alternative lower bound
  $c(S) \geq P + \pv(Q(\bar S))
  - \sum_{t \in Q(\bar S)} l_t
  = P - \exc(Q(\bar S)) \text{.}$
\end{proof}

\lemmafeasibility*
\begin{proof}
  If $\opdef(Q) > 0$ for some $Q$, then some upper bound $m_t$ cannot be met.
  If $\exc(Q) > 0$ for some $Q$, then some lower bound $l_t$ cannot be met.
  For the direction from right to left, consider an infeasible scheduling instance with lower and upper bounds.
  By Lemma~\ref{lemma:flow_feasibility} we have that the maximum flow $f$ for this instance has value $|f| < P$.
  Hence, there must be an $\alpha$-$\omega$ cut $(S, \bar S)$ of capacity $c(S) < P$.
  Lemma~\ref{lemma:cut} now implies that $\opdef(Q(S)) > 0$ or $\exc(Q(\bar S)) > 0$.
\end{proof}

\lemmainvariant*
\begin{proof}
  We show Invariants 1 and 2 via structural induction on the realigned schedule $S_{\real}$.
  Then we show that Invariant 2 implies Invariant 3.
  For the induction base, consider $S_{\aug}$, let $V$ be an arbitrary valley in $S_{\aug}$ with $c \coloneqq \crit(V) \geq 2$, and let $C$ be the critical set with $C \sim V$.
  We must have $k_V \leq c$, otherwise $V$ would contain a full busy interval on processor $k_V > c$ and hence also an engagement $t \in V$ of processor $k_V$, which by construction of $S_{\aug}$ would have $\crit(t) = k_V > c$.
  This is a direct contradiction to $\crit(V) = \max_{t \in V} \crit(t) = c$.
  Invariant 2 now follows since by construction of $S_{\aug}$ and our choice of $C$ we have for every $t \in C$ that processors $1, \ldots, k_V, \ldots, c$ are busy at $t$.
  For Invariant 1, let $D$ be a section of $C$ with $\phi(C \cap D) \leq k_V - \delta$ for some $\delta \in \mathbb{N}$.
  With $k_V \leq c$ we get $\phi(C \cap D) \leq c - \delta$ and hence by Lemma~\ref{lemma:valley}, we have that the left valley $V_l$ or the right valley $V_r$ of $C$ and $D$ exists and $|J(V_l)| + |J(V_r)| \geq \delta$.
  With the initial definition of the supply $\res(V)$ of a valley, we get the desired lower bound of $\res(V_l) + \res(V_r) \geq 2 \delta$.

  Now suppose that Invariants 1 and 2 hold at all steps of the realignment up to a specific next step.
  Let $V$ again be an arbitrary valley of $\crit(V) \geq 2$ and let $k$ be the processor currently being filled.
  Let furthermore $k_V, k'_V$ be the critical processor for $V$ before and after, respectively, the next step in the realignment.
  There are four cases to consider for this next step.

  \paragraph{Case 1:}
      Some $V' \supseteq V$ is closed on processor $k$.
      Then no valley $W$ intersecting $V$ has been closed so far on $k$.
      Also, since $\close(k, \cdot)$ only moves the busy slot of the highest busy processor below $k$, we know that the stair property holds within $V$ when only considering processors $1, \ldots, k$.
      We show that the closing of $V'$ on $k$ reduces the critical processor of $V$ by at least $1$, i.e. $k'_V \leq k_V - 1$.
      If $k_V = k$, then $V' \supseteq V$ is closed on processor $k_V$ and hence by definition we have $k'_V \leq k_V - 1$.
      If $k_V < k$, suppose for contradiction that $k_V \leq k'_V \leq k$, where $k'_V \leq k$ again holds by definition of $k'$ since $V' \subseteq V$ is closed on processor $k$.

      Let $B \subseteq V$ be a full busy interval on $k_V$ before the close of $V'$.
      We show that $B \subset V$, i.e.\ that there must be some $t \in V$ idle on $k_V$ before the close.
      The stair-property then implies that processors $k_V, \ldots, k$ are idle at $t$ before the close.
      Since some $V' \supseteq V$ is closed, clearly $V \subset T$ by the choice of $V'$ as valley of some critical set in the realignment definition.
      Therefore we have $\min V - 1 \in T$ or $\max V + 1 \in T$, without loss of generality we assume the former.
      We show that $t \coloneqq \min V - 1$ must be busy on processor $k_V$ before the close.
      Let $W$ be the valley with $W \sim t$ and $t \in W$.
      We know that $W \supseteq V$ since $V$ is a valley and hence $V \prec t \sim W$.
      By our case assumption and the definition of the realignment, no $W' \supseteq W$ can have been closed on processor $k$ so far.
      With $W \supseteq V$ and the definition of $k_W$ we get $k_W \geq k_V$, where $k_W$ is the critical processor of $W$ before the close.
      Our induction hypothesis now implies that processors $1, \ldots, k_V, \ldots, k_W$ are busy at $t$ before the close.
      For $B \subseteq V$ to be a (full) busy interval on $k_V$ before the close, we hence must have $\min V \notin B$.
      We know by definition of the realignment and the subroutine $\close$ that for every $k'$ with $k_V \leq k' < k$ and every $t \in V$:
      \begin{itemize}
        \item
          If $t$ was idle on $k'$ before the close, then $t$ is still idle on $k'$ after the close (definition of $\close$, $k' < k$).
        \item
          If $t$ was idle on $k_V$ before the close, then $t$ was idle on $k'$ before (stair-property with $k_V \leq k'$) and hence $t$ is still idle on $k'$ after the close.
        \item
          If $t$ was part of a full busy interval $B \subseteq V$ on $k_V$ before the close, then $t$ was idle on $k_V + 1$ before the close.
          Otherwise, by the stair property there would have been a full busy interval $B' \subseteq B \subseteq V$ on processor $k_V + 1 \leq k$ before the close, contradicting the definition of $k_V$.
          Hence $t$ was idle on $k$ before by the stair-property and therefore $t$ is idle on $k_V$ after the close (by the definition of close).
      \end{itemize}
      Taken together, for $t \in V$ to be busy on $k'$ after the close, $t$ must have been busy on $k'$ before the close (definition $\close, k' < k$) and $t$ cannot have been part of a full busy interval $B \subseteq V$.
      Hence $t \in B$ for some \emph{partial} busy interval $B \subseteq V$ on $k'$ before the close.
      For $B' \subseteq V$ to be a full busy interval on $k'_V$ after the close (with $k_V \leq k'_V < k$), we must have $B' \subseteq B$, as shown in the following sketch.
      \begin{figure}[H]
        \centering
        \includegraphics[scale=0.4]{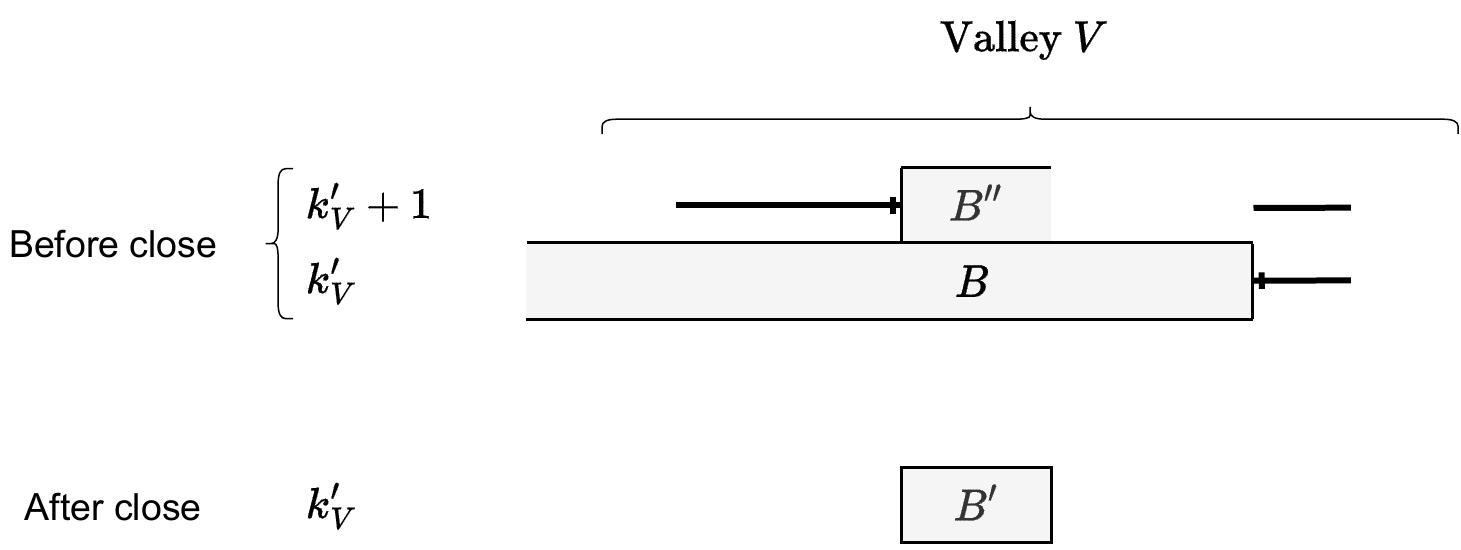}
      \end{figure}
      Hence there must have been a busy interval $B'' \subseteq [\min B', \max B]$ on processor $k'_V + 1 > k_V$ before the close, which contradicts the choice of $k_V < k$.
      In conclusion, we have $k'_V \leq k_V - 1$, which allows us to prove Invariants 1 and 2.
      If $\phi(C \cap D) \leq k'_{V} - \delta$ for some $\delta \in \mathbb{N}$ and some section $D$ of $C$, then $\phi(C \cap D) \leq k_V - (\delta + 1)$ and hence by induction hypothesis the left valley $V_l$ or the right valley $V_r$ for $C, D$ exists and $\res(V_l) + \res(V_r) \leq 2 (\delta + 1)$ both before and after the close.
      Our induction hypothesis also implies that for every $t \in C \cap V$, processors $1, \ldots, k_V$ are busy before the close.
      Since at most the uppermost busy slot is moved by $\close$, after the close of $V'$ we still have that processors $1, \ldots, k_V - 1 \geq k'_V$ are busy.

  \paragraph{Case 2:}\label{case:2}
      Some $V' \subset V$ is closed on processor $k$.
      Again, no $V'' \supseteq V$ can have been closed on processor $k$ so far.
      We show that $k_V = k \geq k'_V$, i.e.\ that the critical processor of $V$ before the close of $V'$ is the processor currently being filled.
      Let $W$ be the valley for which $V'$ is closed, i.e.\ $V'$ is closed during $\fillop(k, W)$.
      We must have $W \supset V'$ and therefore no $W' \supseteq W$ has been closed on $k$ so far.
      Also, for $V'$ to be closed in $\fillop(k, W)$, there must be some busy interval $B \subseteq W$ on $k$ before the close, hence $k_W = k$.
      Since $V' \subset V$ and $V' \subset W$, $V$ and $W$ intersect ($V' \neq \emptyset$ by definition of $V'$ as valley).
      Let $C_W$ be the critical set with $C_W \sim W$.
      If $V \prec W$, then by choice of $V'$ as valley of $C_W$ we must have $V \subseteq V'$, which contradicts our case assumption.
      Therefore $V \succeq W$ and $V \supseteq W$, which in turn implies $k_V \leq k_W = k$.
      Since processor $k+1$ is already completely filled before the close, we have $k_V = k \geq k'_V$.

      For Invariant 1, again let $\phi(C \cap D) \leq k'_V - \delta$ and hence $\phi(C \cap D) \leq k_V - \delta$ for some $\delta \in \mathbb{N}$ and some section $D$ of $C$.
      Our induction hypothesis implies that the left valley $V_l$ or the right valley $V_r$ of $C, D$ exists and that both before and after the close we have $\res(V_l) + \res(V_r) \geq 2 \delta$.
      For Invariant 2, observe that $V' \cap C = \emptyset$ since our case assumption $V' \subset V$ implies $V' \prec C$.
      Therefore, no slots of $C$ are modified when $V'$ is closed.
      Invariant 2 now directly follows from the induction hypothesis and $k'_V \leq k_V$.

  \paragraph{Case 3:}
      Some $V'$ with $V' \cap V = \emptyset$ is closed on processor $k$.
      We first show that $\min V - 1 \notin V'$ and symmetrically $\max V + 1 \notin V'$.
      Consider $t \coloneqq \min V - 1$ and assume $t \in T$.
      By choice of $V$ and $t$ we must have $t \succ V$.
      If $t \in V'$, we would have $V' \succ V$ and hence $V' \supseteq V$, which contradicts our case assumption.
      Symmetrically, we know that $\max V + 1 \notin V'$.
      Therefore the close of $V'$ does not modify the schedule within $[\min V - 1, \max V + 1]$, implying that no partial busy interval in $V$ before the close can become a full busy interval.
      Hence we have $k_V = k'_V$ and Invariants 1 and 2 follow as in Case 2.

  \paragraph{Case 4:}
      The call to $\fillop(k, T)$ returns and $\res(V')$ is decreased by 1 for every valley $V'$ such that some valley intersecting $V'$ has been closed during $\fillop(k, T)$.
      First observe that the schedule itself does not change by this step but processor $k$ is now fully filled, which implies $k'_V \leq k_V$.
      Invariant 2 then follows directly from the induction hypothesis.
      We consider two subcases.
      If during $\fillop(k, T)$, no valley $V'$ intersecting $V$ was closed on $k$, then $\res(V)$ does not change and Invariant 1 follows from the induction hypothesis and $k'_V \leq k_V$.

      If on the other hand some valley $V'$ intersecting $V$ was closed on $k$ during $\fillop(k, T)$, then $\res(V)$ is decreased by $1$ to $\res'(V) \coloneqq \res(V) - 1$.
      As argued in Cases 1 to 3, the critical processor of $V$ decreases monotonically during $\fillop(k, T)$.
      Consider the schedule right before the first valley $V'$ intersecting $V$ is closed on $k$.
      Let $k^0_V$ be the critical processor for $V$ at this point of the realignment and $k^1_V$ the critical processor right after $V'$ is closed.
      We have $k'_V \leq k^0_V - 1$:
      If $V' \supseteq V$, then as argued in Case 1, we have $k^1_V \leq k^0_V - 1$ and hence $k'_V \leq k_V \leq k^1_V \leq k^0_V -1$.
      If $V' \subset V$, then as argued in Case 2 we have $k^0_V = k$.
      Since by our case assumption $\fillop(k, T)$ returns in the next step, we have $k'_V \leq k - 1$ and hence $k'_V \leq k^0_V -1$.
      Invariant 2 now follows by our induction hypothesis.
      If $\phi(C \cap D) \leq k'_V - \delta$ then $\phi(C \cap D) \leq k^0_V - (\delta + 1)$ and hence by our induction hypothesis the left valley $V_l$ or the right valley $V_r$ of $C, D$ exists and before the close we have $\res(V_l) + \res(V_r) \geq 2 (\delta + 1)$.
      Since $\res$ is decreased for every valley by at most $1$, we have after the close that $\res'(V_l) + \res'(V_r) \geq 2 \delta$.

  We conclude by showing that Invariant 2 implies Invariant 3.
  Let $V$ be an arbitrary valley during the realignment of $S_{\aug}$ and $B \subseteq V$ a busy interval on processor $k_V$ with $B \sim V$.
  Let $C$ be the critical set with $C \sim V$.
  Note that $B \sim V$ implies that $B$ intersects $C$.
  Assume for contradiction that $B$ is not a section of $C$.
  Then $\min B$ lies strictly within a subinterval of $C$ or symmetrically $\max B$ lies strictly within a subinterval of $C$.
  We assume the first case, i.e.\ $t\coloneqq \min B - 1 \in C$ and $\min B \in C$.
  The second case follows by symmetry.
  If $t \in V$, then time slot $t$ is busy on processor $k_V$ by Invariant 2.
  Therefore, $B$ cannot be a (full) busy interval on processor $k_V$, contradicting the choice of $B$.
  If $t = \min V-1$, then consider the valley $W$ with $t \in W$ and $t \sim W$ and let $C_W$ be the critical set with $C_W \sim W$.
  We must have $W \supset V$, $W \succ V$ and $t \in C_W$.
  Therefore $k_W \geq k_V$ and Invariant 2 implies that $t = \min B - 1$ is busy on processor $k_V$, again contradicting the choice of $B$ as full busy interval on processor $k_V$.
\end{proof}

\lemmasreal*
\begin{proof}
  Since in the while-loop of $\fillop(k, V)$, the busy interval $B \subseteq V$ on $k_V$ always is a section of $C$ if $V \sim C$ (Invariant 3), the left valley $V_l$ and the right valley $V_r$ of the critical set $C$ and interval $B$ are properly defined.
  Also since $\hat \phi(B) \leq k - 1$, Invariant 1 implies that $V_l$ or $V_r$ exists and that there is sufficient $\res$ such that one of the two valleys of $C$ is closed in this iteration.
  This reduces the number of idle intervals on processor $k$ by at least $1$, since Invariant 2 implies that $V_l$ or $V_r$ cannot end strictly within an idle interval on $k$.
  Hence all terms in the realignment are well defined and the realignment terminates.
\end{proof}

\end{document}